\documentclass[a4paper,fleqn]{cas-sc}

\usepackage[numbers,sort&compress]{natbib}

\usepackage{tikz}
\usetikzlibrary{positioning,arrows.meta,shapes.geometric,fit,calc,backgrounds}
\usepackage{booktabs}
\usepackage{tabularx}
\usepackage{listings}
\usepackage{xcolor}
\usepackage{amsmath,amssymb}
\usepackage{enumitem}

\definecolor{kw}{HTML}{1565C0}
\definecolor{cmt}{HTML}{6A6A6A}
\definecolor{str}{HTML}{0F9D58}
\lstdefinestyle{p11py}{%
  language=Python,
  basicstyle=\ttfamily\scriptsize,
  keywordstyle=\color{kw}\bfseries,
  commentstyle=\color{cmt}\itshape,
  stringstyle=\color{str},
  numbers=none,
  showstringspaces=false,
  breaklines=true,
  columns=fullflexible,
  xleftmargin=0.4em,
  aboveskip=0.4em,
  belowskip=0.4em,
  frame=tb,
  framesep=0.3em,
  framerule=0.4pt,
  rulecolor=\color{black!30},
}

\newcommand{\pcode}[1]{\texttt{\small #1}}

\hypersetup{colorlinks=true, linkcolor=black, citecolor=black,
            urlcolor=black, filecolor=black}
\usepackage{placeins}

\usepackage{amsthm}
\newtheoremstyle{p11plain}%
  {6pt}{6pt}{\itshape}{0pt}{\bfseries}{.}{ }{}
\newtheoremstyle{p11definition}%
  {6pt}{6pt}{\normalfont}{0pt}{\bfseries}{.}{ }{}

\theoremstyle{p11definition}

\theoremstyle{p11plain}
\newtheorem{theorem}{Theorem}

\newtheorem{proposition}[theorem]{Proposition}
\renewcommand{\proofname}{Proof}
\makeatletter
\renewenvironment{proof}[1][\proofname]{\par
  \pushQED{\qed}\normalfont \topsep6\p@\@plus6\p@\relax
  \trivlist
  \item[\hskip\labelsep\bfseries
    #1\@addpunct{.}]\ignorespaces
}{\popQED\endtrivlist\@endpefalse}
\makeatother

\begin{document}
\let\WriteBookmarks\relax

\shorttitle{Audit-First Rollback Semantics}
\shortauthors{X.\ Qin et al.}

\title[mode=title]{Audit-First Rollback Semantics for Safety-Critical
       Deployment Pipelines}

\author[1]{Xue Qin}[orcid=0009-0009-3642-2663]
\credit{Conceptualization, Methodology, Software, Formal analysis, Investigation, Writing -- original draft}

\author[2]{Simin Luan}[orcid=0000-0003-1138-1892]
\credit{Investigation, Validation, Writing -- review and editing}

\author[3]{Cong Yang}[orcid=0000-0002-8314-0935]
\cormark[1]
\ead{cong.yang@suda.edu.cn}
\credit{Supervision, Writing -- review and editing}

\author[2]{Zhijun Li}[orcid=0000-0001-9129-9957]
\cormark[1]
\ead{lizhijun_os@hit.edu.cn}
\credit{Supervision}

\affiliation[1]{organization={School of Software, Harbin Institute of Technology},
            city={Harbin},
            country={China}}

\affiliation[2]{organization={School of Computer Science and Technology, Harbin Institute of Technology},
            city={Harbin},
            country={China}}

\affiliation[3]{organization={School of Future Science and Engineering, Soochow University},
            city={Suzhou},
            country={China}}

\cortext[1]{Corresponding authors}

\begin{abstract}
Distributed deployment runtimes carry a coherence obligation that
classical fault-tolerance frameworks do not name directly: the
\emph{live state} a component is configured to run and the
\emph{audit chain} that records how it got there must agree at
every terminal configuration. Prior works mainly focus on
individual aspects of the deploy-time fault surface (canary
controllers, configuration rollback, signed attestations), leaving
the cross-cutting question of audit/live coherence under fail-stop
crash only loosely specified. Yet a key systems question remains
unresolved: how can a deployment runtime guarantee that the audit
chain answers truthfully about live state even when a transition
crashes mid-flight? We present \emph{audit-first rollback
semantics}, a fault-tolerance mechanism that guarantees audit/live
coherence at every committed terminal under fail-stop crashes
during transition phases. The mechanism pairs with
\emph{provisional state machines}, pipelines whose ``active but
not yet promoted'' states carry an explicit rollback contract and
a bounded deadline. We instantiate both in a runtime deployment
system and run a dependability evaluation against a fail-open
variant of the same pipeline across twelve fault-injection points
spanning three structural failure classes. Across $1{,}200$
injected-failure trials, audit-first rollback achieves $100\%$
audit/live-state coherence ($600/600$, Wilson 95\% CI
$[0.994,1.000]$) with per-cell $p_{95}$ recovery latency below
$500$~ms (SLO PASS in $12/12$ cells); the fail-open variant
preserves coherence on only $33\%$ of trials ($200/600$, Wilson
95\% CI $[0.297,0.372]$). We further lift the construction to a
cross-bridge coordination protocol with a safety argument for
fleets of fail-stop bridges, leaving fleet-scale empirical
evaluation to follow-on work.
\end{abstract}


\maketitle

\section{Introduction}\label{sec:intro}

A deployment pipeline reaches its provisional-active phase: the new
version of a capability is live in the runtime's active-version map,
traffic is being served against it, but the audit chain has not yet
recorded the deployment as terminal-promoted. At this moment, the
system crashes; say, the metrics provider raises a \pcode{TypeError}
because of an unexpected timestamp format, or a transient store
update fails. Two questions arise, with two incompatible answers:
\emph{What version is running?} Live state says the new version.
\emph{What version did the pipeline successfully promote?} The
audit chain says none; the last terminal record is
\pcode{SHADOW\_PASSED}, not \pcode{PROMOTED}. These diverge. Which
one wins?

The divergence is a coherence failure between two views of the
same component, and the recovery posture chosen by the runtime
decides which view is trusted. Mainstream cloud deployment
runtimes resolve disagreement by trusting live state: the partial
deploy is allowed to keep serving on the assumption that some
throughput is better than none. We refer to this resolution
posture as \emph{fail-open}; the operator is notified and rolls
forward or back at leisure. For high-availability web services
the posture is well-justified, because the cost of unavailability
typically exceeds the cost of running a half-deployed configuration.
The same posture is structurally wrong for any deployment runtime
whose audit chain feeds downstream safety-critical machinery,
because the audit chain becomes the basis for retrospective
reasoning about what the system actually did.
The framing follows the dependability taxonomy
of~\citet{avizienis2004dependability}: the property at stake is
\emph{integrity} of the audit chain (it answers truthfully) and
\emph{availability} of the live state (the deployed component
runs the right version), with the fault-mode in question being
unhandled fail-stop crash during a provisional transition.

For safety-critical distributed systems the calculus inverts.
An embodied robot whose deployment runtime trusts a stale live
state may command unsafe motions; a clinical pipeline whose
deployment runtime trusts a stale live state may compute on the
wrong model version while reporting that the new one is live.
In both cases the audit chain is the spine of retrospective
safety analysis: regulators, operators, and downstream
data-consuming systems reason about what happened through audit
records. A pipeline that lets live state outrun the audit chain
breaks this contract by construction. The structural property
shared across these settings is that the audit chain must remain
truthful about live state at every committed terminal even when
the transition crashes mid-flight, under the standard fail-stop
fault model of~\citet{schlichting1983failstop}.

We argue that for safety-critical agents the correct default is the
opposite: \emph{audit-first rollback semantics}. When the pipeline
encounters any failure during a provisional transition, the live
state is reverted to match the audit chain's last consistent point
before any terminal record is written. The audit chain becomes the
source of truth, and live state is required to converge to it.

The construction has two coupled parts. First, \emph{provisional
state machines} explicitly partition pipeline states into committed
(where live state and the audit chain agree) and provisional (where
live state has flipped but the audit chain has not yet recorded a
terminal). Every provisional state ships with a rollback closure
that knows how to revert live state, and every provisional state
has a bounded lifetime. Second, \emph{audit-first rollback} dictates
that on any failure inside a provisional region, the rollback
closure runs first; the audit chain records a terminal state only
after the closure has attempted to converge live state, and a
closure failure is itself audited so the operator knows when
divergence persists.

\paragraph{Contributions} (1) \emph{Audit-first rollback
semantics} as a deliberate, named fault-tolerance mechanism for
distributed deployment runtimes whose live-state and audit-chain
views can diverge under fail-stop crash. (2) \emph{Provisional
state machines} as a state-machine pattern with explicit
rollback contract and bounded deadline, distinct from two-phase
commit because the second phase is \emph{time-driven} (a canary
window or bounded soak), not \emph{coordination-driven}
(participant votes). (3) A dependability evaluation that
exercises a runtime implementation across twelve fault-injection
points under two postures and $1{,}200$ injected-failure trials,
reporting audit/live-state coherence rates with Wilson 95\% CIs
and per-cell $p_{95}$ recovery latency against an explicit SLO.
(4) A multi-bridge lifting of the construction
(\S\ref{sec:multibridge}): a cross-bridge coordination protocol
under bounded async with two new safety propositions
(per-bridge truthfulness and eventual fleet convergence),
showing that the per-bridge invariants compose without
strengthening the per-bridge fault model.

We do not claim novelty for state machines, transactional
rollback, or the canary-deployment pattern. The contribution is
the named source-of-truth choice, the mechanical implementation
pattern, the dependability measurement under fault injection,
and the multi-bridge safety lift.

\section{Motivation}\label{sec:motivation}

\subsection{Two Failure Modes That Diverge Live State from Audit}

Two real failure modes motivate the construction. Both arise
naturally in any pipeline that does a provisional live-state flip
before reaching a terminal audit record.

\paragraph{Failure mode A: metrics provider raises mid-canary}
The canary stage of a deployment pipeline polls a metrics provider
to compute a success rate over the canary window. The provider
walks an execution store and filters records by start time. A
regression in an unrelated module begins emitting timezone-naive
ISO timestamps; the comparison against the canary's timezone-aware
start time raises \pcode{TypeError}. The pipeline has already
provisionally promoted the new version into the active-version map
(so traffic is being served against it) when the metric loop
crashes. Without audit-first rollback the pipeline records the job
as \pcode{FAILED}, but live state still has the new version
installed. An operator inspecting the audit chain sees ``deploy
failed''; an operator inspecting the active-version map sees
``deploy succeeded.''

\paragraph{Failure mode B: rollback function itself raises} The
rollback closure reverts the active-version map by removing or
overwriting the entry for the upgraded capability. In a rare race
where the map has been mutated externally between the provisional
promote and the rollback (for example, a concurrent upgrade
affecting an overlapping capability set), the closure itself
raises. The pipeline records \pcode{ROLLED\_BACK} because the
audit-write happens after the call returns successfully, except
the call did not return successfully. The audit chain says
rolled-back; live state has the new version still installed.

Both modes share a structural property: the live-state mutation and
the audit-chain write are not atomic, and there is a window in
which a crash leaves them disagreeing. Any provisional region in
any pipeline has this window unless an explicit rule names the
source of truth and a mechanism enforces convergence.

\subsection{Why Cloud-Default Fail-Open Is Wrong for Embodied Agents}

The cloud deployment community has converged on fail-open as a
sensible default because the cost of partial-deploy-running
typically exceeds the cost of unavailability. A web service running
half its pods on the new version is still serving requests; the
operator decides at leisure whether to roll forward or back.

For embodied robots, the same partial deploy is a robot running a
half-deployed grasp capability. Wrong arm motion is not a tolerable
failure mode while the operator considers options. Two additional
considerations apply. \emph{Episodic memory consistency.} The audit
chain is the spine of the agent's recorded life history. A planner
that consults the episodic memory expects the recorded version
field to match the version that actually executed; sustained
divergence silently poisons every record laid down during the
divergence window, and these records can outlive any operator
reconciliation. \emph{Operator mental model.} Operators of
regulated robots reason about the system through the audit chain
because the audit chain is what regulators and customers will
eventually inspect. A pipeline that lets live state outrun the
audit chain breaks this contract by construction.

The default for safety-critical agents must therefore invert: on
any failure, live state must be reverted to match the audit chain,
and the convergence itself must be audited.

\section{Construction}\label{sec:construction}

\subsection{Provisional State Contract}\label{sec:provisional-contract}

A pipeline is a state machine $(V, E)$ with a designated start
state and a set of terminal states. We partition $V$ into
\emph{committed states}, where live state and the audit chain
agree (terminal states are committed by definition; some
intermediate states are also committed when the transition into
them does not mutate live state), and \emph{provisional states},
where live state has been mutated but the audit chain has not yet
recorded a terminal record reflecting the new state.

Every provisional state $p$ ships with two annotations:
$\textsc{rollback}(p)$, a closure with no required arguments that
reverts live state to the most recent committed state's snapshot;
and $\textsc{deadline}(p)$, a bounded duration after which the
provisional state must have transitioned to a committed state or
the rollback closure must run.

The pipeline contract has three obligations: \emph{initiation},
every transition into a provisional state must construct and install
the rollback closure on the job before any live-state mutation, with
the closure capturing the pre-transition snapshot it needs;
\emph{bounded duration}, every provisional state has an outgoing
transition triggered no later than its deadline, the triggering
transition being either to a committed terminal state (success
path) or to a rollback (failure path); and \emph{no leakage}, every
reachable provisional state has at least one outgoing edge to a
committed state, and the pipeline cannot end in a provisional state.

\subsection{Audit-First Rollback Semantics}\label{sec:semantics}

The contract above does not specify what happens when the
provisional transition itself raises an exception. \emph{Audit-first
rollback} specifies it. When the pipeline encounters any exception
inside a provisional state, the implementation must (1)~catch the
exception in a guard wrapping the provisional region; (2)~run the
rollback closure synchronously (or with a bounded timeout); the
closure is best-effort and may itself raise; (3)~record the
rollback in the audit chain; (4)~transition to a committed terminal
state, \pcode{ROLLED\_BACK} if the closure ran without raising,
\pcode{FAILED} if the closure itself raised.

The order matters. The audit-write happens \emph{after} the
rollback attempt, not concurrently and not before. This gives the
construction a useful property: the audit chain's terminal record
reflects what actually happened to live state. A \pcode{ROLLED\_BACK}
record carries the implicit guarantee that the live-state revert
succeeded; a \pcode{FAILED} record signals the operator that live
state and audit chain may have diverged and that human
reconciliation is required. The operator never has to guess which
side of the divergence is authoritative; the audit chain answers,
and answers truthfully, because it is written last and only with
knowledge of the rollback's outcome.

We call this \emph{audit-first} even though the audit-write happens
last in clock time, because the audit chain wins the \emph{semantic}
race: when the audit chain says \pcode{ROLLED\_BACK}, the system
behaves as if the deploy never happened. Live state is forced to
agree with the audit chain, not the other way round.

\subsection{Implementation Pattern}\label{sec:pattern}

The implementation pattern is small and mechanical: a try/except
guard wrapping the provisional region. The guard's try block
performs the provisional live-state mutation and then calls the
body that operates within the provisional region. The except block
calls the rollback closure, then records the terminal audit row. A
nested try/except inside the rollback path distinguishes
\pcode{ROLLED\_BACK} from \pcode{FAILED}. In our implementation the
canary stage of the deployment pipeline is the canonical
provisional region; the relevant code is \pcode{\_run\_canary} in
\pcode{evolution\_pipeline.py} lines~654--752, with the schematic
shown below.

\begin{lstlisting}[style=p11py]
async def _run_canary(self, job):
    self._promote(job.capability, job.to_version,
                  job.from_version)  # provisional flip
    try:
        return await self._run_canary_body(job, canary_start_dt)
    except Exception as exc:
        try:
            self._rollback(job.capability, job.to_version,
                           job.from_version,
                           f"canary internal error: {exc!r}")
        except Exception as rb_exc:
            await self._store.update(
                job.job_id, status=JobStatus.FAILED,
                error=f"canary crash + rollback failure: "
                      f"{exc!r}; rb={rb_exc!r}")
            return await self._fetch_or_raise(job.job_id)
        await self._store.update(
            job.job_id, status=JobStatus.ROLLED_BACK,
            rollback_reason=str(exc))
        return await self._fetch_or_raise(job.job_id)
\end{lstlisting}

Three properties of this pattern matter for auditability. First,
the guard's scope is exactly the post-provisional region: the
\pcode{\_promote} call sits \emph{outside} the try block, so a
failure in the provisional flip itself cannot trigger a rollback of
state that was never installed. Second, the rollback closure is
called \emph{before} the store update, so the terminal record
reflects the rollback's outcome rather than predicting it. Third,
the \pcode{FAILED} branch's error string carries both the original
exception and the rollback exception, so the audit chain preserves
the full causal chain.

\subsection{State-Machine Guarantees}\label{sec:guarantees}

The construction guarantees three properties. \emph{Termination}:
every provisional state has at least one outgoing edge to a
committed terminal state (\pcode{PROMOTED}, \pcode{ROLLED\_BACK},
or \pcode{FAILED}); combined with the bounded-deadline obligation,
no execution can dwell in a provisional state indefinitely.
\emph{Live-state coherence at committed states}: at every committed
state reachable in normal execution, the live-state view of the
upgraded capability matches the audit chain's most recent record;
\pcode{ROLLED\_BACK} and \pcode{PROMOTED} both carry this guarantee,
\pcode{FAILED} explicitly does not, and a \pcode{FAILED} record is
exactly the signal that operator reconciliation is required.
\emph{Crash safety inside the guard}: any exception raised inside
the provisional region is caught and routed through the rollback
path. One scenario sits outside these guarantees, a bridge crash
mid-rollback where the host process terminates between the rollback
closure call and the store update; we discuss this in
Section~\ref{sec:discussion}.

\begin{figure}[pos=t!]
  \centering
  \includegraphics[width=\linewidth]{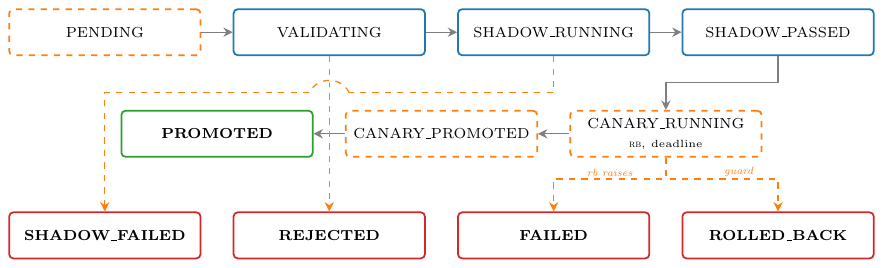}
  \caption{The eight-state deployment pipeline. Committed states
    (solid nodes) and provisional states (dashed nodes) are
    partitioned per Section~\ref{sec:provisional-contract}. Each
    provisional state carries a rollback closure (\textsc{rb}) and
    a bounded deadline. Failure edges through the
    Section~\ref{sec:semantics} guard converge live state with the
    audit chain at \pcode{ROLLED\_BACK} (closure succeeded) or
    flag persistent divergence at \pcode{FAILED} (closure itself
    raised).}
  \label{fig:state-machine}
\end{figure}

\section{Informal State-Machine Specification}\label{sec:formal}

We now restate Section~\ref{sec:construction}'s informal contract as
a state-machine specification, in the sense
of~\citet{schneider1990smr}'s state-machine approach to
fault-tolerant services, adapted here from $N$-replica replication
to single-process audit/live coherence. The specification is
deliberately small and presented informally: it is the smallest
description that captures the audit-first rollback rule and that
lets us state the correctness arguments against the implementation
sketched in Section~\ref{sec:pattern}. Mechanised refinement
checking~\citep{lynch1989forward,lamport1994tla} against a TLA+
or P model is a natural next step but is out of scope for this
paper; the propositions below carry informal justifications, not
machine-checked proofs.

\subsection{Specification}\label{sec:formal-spec}

Let $\mathcal{V}$ be a finite set of pipeline states with a
distinguished start state $v_0$ and a non-empty set
$\mathcal{V}_T \subseteq \mathcal{V}$ of committed terminal states.
Let $\mathcal{V}_P \subseteq \mathcal{V} \setminus \mathcal{V}_T$
be the set of provisional states. The pipeline is a labelled
transition system $(\mathcal{V}, \mathcal{V}_P, \mathcal{V}_T,
\Sigma, \to)$ where $\Sigma$ partitions transitions into
$\Sigma = \Sigma_{\mathrm{progress}} \uplus
\Sigma_{\mathrm{rollback}} \uplus \Sigma_{\mathrm{failure}}$.
Each transition is annotated with a pair of stateful effects:
\emph{live}, the actuator-side mutation, and \emph{audit}, the
append to the audit chain.

A configuration is a triple $(v, L, A)$ where $v \in \mathcal{V}$
is the current state, $L$ is the live-state value, and $A$ is the
audit chain (a sequence of records). The initial configuration is
$(v_0, L_0, \langle \rangle)$. A run is a sequence of
configurations connected by transitions. Each transition
$(v, L, A) \xrightarrow{\sigma} (v', L', A')$ obeys:

\begin{itemize}[leftmargin=1.2em,itemsep=0pt,topsep=2pt]
  \item If $\sigma \in \Sigma_{\mathrm{progress}}$ and $v' \in
    \mathcal{V}_P$ then $L'$ is the post-provisional value
    determined by $\sigma$'s \emph{live} effect, and the rollback
    closure $\textsc{rollback}(\sigma)$ is recorded as a job-local
    field; $A'=A$ (no audit append before the provisional region
    terminates).
  \item If $\sigma \in \Sigma_{\mathrm{progress}}$ and $v' \in
    \mathcal{V}_T$ then $L' = L$ (no further live mutation) and $A'
    = A \mathbin{\cdot} r_{v'}$ where $r_{v'}$ is the terminal
    audit record.
  \item If $\sigma \in \Sigma_{\mathrm{rollback}}$ then $L' =
    \textsc{rollback}(\sigma)(L)$, $A' = A \mathbin{\cdot}
    \pcode{ROLLED\_BACK}_\sigma$. Precondition: $v \in
    \mathcal{V}_P$.
  \item If $\sigma \in \Sigma_{\mathrm{failure}}$ then $L'$ is the
    result of best-effort applying $\textsc{rollback}(\sigma)$ to
    $L$; if the closure raised, then $L'$ is unspecified (the
    audit chain will record the divergence) and $A' = A
    \mathbin{\cdot} \pcode{FAILED}_{\sigma}$; else $L' =
    \textsc{rollback}(\sigma)(L)$ and $A' = A \mathbin{\cdot}
    \pcode{ROLLED\_BACK}_\sigma$.
\end{itemize}

We are interested in three properties.

\begin{proposition}[Termination]\label{prop:term}
Every run of the specification reaches a configuration whose state
lies in $\mathcal{V}_T$, modulo the bounded-deadline obligation of
Section~\ref{sec:provisional-contract}.
\end{proposition}

\begin{proof}[Argument]
The provisional set $\mathcal{V}_P$ is finite. Each provisional
state $p$ has a deadline $\delta(p) < \infty$. By the deadline
obligation, some outgoing transition fires within $\delta(p)$, and
its target is either in $\mathcal{V}_P$ (decreasing remaining
budget) or in $\mathcal{V}_T$. Since $\mathcal{V}_P$ is finite and
$\sum_p \delta(p) < \infty$, the run cannot dwell in $\mathcal{V}_P$
indefinitely.
\end{proof}

\begin{proposition}[Live-state coherence at terminals]\label{prop:coh}
Let $(v, L, A)$ be any reachable configuration with $v \in
\mathcal{V}_T$ and let $r$ be the last audit record in $A$. Then
either (i) $r = \pcode{PROMOTED}$ and $L$ equals the post-deploy
value, or (ii) $r = \pcode{ROLLED\_BACK}$ and $L$ equals the
pre-deploy value, or (iii) $r = \pcode{FAILED}$ (and the operator
must reconcile $L$ manually).
\end{proposition}

\begin{proof}[Argument]
Case~(i) follows directly from the
$\sigma \in \Sigma_{\mathrm{progress}}, v' \in \mathcal{V}_T$
clause: $L$ is unmodified at the audit-write step. Case~(ii)
follows from the $\Sigma_{\mathrm{rollback}}$ and the success branch
of $\Sigma_{\mathrm{failure}}$: $L'$ is set by the rollback closure
before the audit record is appended. Case~(iii) is the failure
branch of $\Sigma_{\mathrm{failure}}$, which we explicitly leave
unconstrained on $L'$ and audit explicitly as
\pcode{FAILED}.
\end{proof}

\begin{proposition}[Crash-safety inside the guard]\label{prop:crash}
Suppose the implementation throws an exception while executing any
transition $\sigma$ whose source is a provisional state. Then the
next configuration is either (a) a \pcode{ROLLED\_BACK}-terminal
configuration with $L$ at the pre-deploy value, or (b) a
\pcode{FAILED}-terminal configuration with $L$ unspecified. In
both cases, the audit chain's terminal record describes the
outcome truthfully.
\end{proposition}

\begin{proof}[Argument]
The implementation pattern of Section~\ref{sec:pattern} matches
$\Sigma_{\mathrm{failure}}$: the outer try/except wraps the
provisional region, the rollback closure is called in the except
handler, and the audit write happens after the closure either
returns or itself raises. Hence the trace of every catching
execution corresponds to the failure branch of the specification.
\end{proof}

\subsection{Refinement Obligation}\label{sec:refinement}

An implementation can be said to satisfy
Propositions~\ref{prop:term},
\ref{prop:coh}, and~\ref{prop:crash} when it refines the
specification: for every implementation trace there is a
specification trace that produces the same audit-chain prefix.
This is the standard simulation relation for refinement of
labelled transition systems~\citep{lynch1989forward}, and is the
programme of which the present specification is one small
instance~\citep{lamport1994tla}. We do not prove refinement
mechanically here; the argument we give in
Propositions~\ref{prop:term}--\ref{prop:crash} is informal.
\citet{desai2013p} and~\citet{newcombe2015use} describe similar
applications of state-machine refinement at industrial scale where
mechanised checking has paid off, and Section~\ref{sec:related}
discusses the broader refinement literature
(IronFleet~\citep{ironfleet2015}, Verdi~\citep{verdi2015},
Mace~\citep{mace2007}) of which a mechanised version of this paper's
specification would be a member.

The proposition cluster permits a useful implementation simplification.
Because crash-safety is established at the specification level
(Proposition~\ref{prop:crash}), the implementation does not need to
maintain a runtime invariant; it only needs to wrap each provisional
region in a try/except guard whose shape matches the
$\Sigma_{\mathrm{failure}}$ branch. This is the source of the small
implementation footprint we report in Section~\ref{sec:eval}.

\section{Implementation}\label{sec:impl}

We implement provisional state machines and audit-first rollback in
a runtime deployment system for embodied robots~\citep{aeros-p4}.
The deployment
pipeline drives capability upgrades on a Franka Panda arm; each
upgrade transitions a single capability's version through an
eight-state pipeline.

\paragraph{Pipeline as canonical example} The pipeline's
\pcode{JobStatus} enum (defined at
\pcode{evolution\_pipeline.py:92-113}) has eight states:
\pcode{PENDING}, \pcode{VALIDATING}, \pcode{SHADOW\_RUNNING},
\pcode{SHADOW\_PASSED}, \pcode{CANARY\_RUNNING},
\pcode{CANARY\_PROMOTED}, terminal \pcode{PROMOTED}, and four
terminal failure states (\pcode{SHADOW\_FAILED}, \pcode{REJECTED},
\pcode{ROLLED\_BACK}, \pcode{FAILED}). Three are provisional in the
sense of \S\ref{sec:provisional-contract}: \pcode{PENDING} (live
state has been earmarked but the validator has not consented),
\pcode{CANARY\_RUNNING} (live state flipped, soak ongoing), and
\pcode{CANARY\_PROMOTED} (transient between canary success and the
audit-chain promotion record). The canary stage is the most
interesting: it is where provisional status persists for a bounded
soak window and where most failure modes manifest.

\paragraph{Audit chain} The audit chain is implemented as the
\pcode{PersistentAgent.episodic\_memory} ring buffer. It is
append-only; each record is a tuple \pcode{(event\_type, intent\_id,
payload)} where \pcode{payload} is a JSON-serialisable dict carrying
the evolution action (\pcode{upgrade}, \pcode{rollback},
\pcode{upgrade\_rejected}), the capability name, the from- and
to-versions, and a free-form reason string. The audit chain has its
own durability story (SQLite write-through, WAL mode); the
deployment pipeline simply appends through the helper
\pcode{\_record\_evolution\_audit} defined in
\pcode{bridge/app/api/routes.py}.

\paragraph{Rollback closure factory} The rollback closure for the
canary stage is constructed by \pcode{\_build\_evolution\_pipeline}
in \pcode{bridge/app/api/routes.py} (lines~1711--1780). The factory
closes over the live-state map (\pcode{state.active\_ecm\_versions}),
the capability name, the to- and from-versions, and the audit-write
helper. When called, the closure overwrites the active-version
entry with the from-version snapshot and emits an
\pcode{action="rollback"} audit record. Constructing the closure at
job-creation time means the snapshot is captured \emph{before} any
provisional mutation; this satisfies the
\S\ref{sec:provisional-contract} initiation obligation.

\paragraph{Per-capability lock} Two concurrent upgrades affecting
the same capability would race on the active-version map. The
pipeline acquires a per-capability \pcode{asyncio.Lock} before
constructing the closure and releases it after the terminal audit
record is written. A second concurrent request returns HTTP~409
(\pcode{Conflict}); operators see the conflict in the
audit-adjacent error log rather than experiencing silent
clobbering.

\section{Evaluation}\label{sec:eval}

The evaluation is a \emph{dependability measurement under fault
injection}, not a throughput or latency benchmark. We measure two
properties along a single grid of injected fail-stop crashes:
(i)~audit/live-state \emph{coherence}, the property
Proposition~\ref{prop:coh} requires the construction to hold at
every committed terminal; and (ii)~per-cell $p_{95}$ recovery
latency against an explicit service-level objective of
$p_{95}\!\le\!500$~ms and $p_{99}\!\le\!1$~s for the
operator-facing rollback signal. Together these answer one
question: does the audit-first runtime, under fail-stop
fault injection, hold its safety property and meet its recovery
SLO simultaneously? We report Wilson 95\% confidence intervals on
the coherence rate and per-cell PASS/FAIL against the SLO.

\subsection{Setup}

We evaluate audit-first rollback against a \emph{fail-open variant}
of the same pipeline. The fail-open variant differs only in the
except branch of the \S\ref{sec:pattern} guard: it omits the
rollback call and records \pcode{FAILED} directly, leaving live
state in whatever provisional configuration it was in at the moment
of the exception. The variants are byte-identical otherwise; we
toggle between them via a fixture flag in the test harness. The
comparison is therefore a \emph{within-construction isolation} of
the audit-first rule (the same pipeline with the rollback rule on
versus off), designed to attribute every observed difference to
the rule itself; it is not a benchmark against external canary
controllers such as Argo Rollouts, Spinnaker, or Flagger
(\S\ref{sec:limitations} names this scope boundary explicitly).

For each crash injection we run $50$ trials per posture
(audit-first vs.\ fail-open), giving $100$ trials per injection
point and $1200$ trials total across the $12$-point grid. Each
trial captures three quantities:
\emph{final job status} (terminal status recorded on the job:
\pcode{ROLLED\_BACK}, \pcode{FAILED}, or, pathologically,
\pcode{CANARY\_RUNNING} if the guard ever leaks a provisional
state); \emph{final live state} (the active-version-map entry for
the capability under test); and \emph{audit chain consistency}
(whether the audit chain's most recent terminal record agrees with
live state). A trial is \emph{consistent} if the audit chain's
most recent record is a terminal one and the live-state value
matches what that record asserts. Inconsistency is exactly the
divergence the construction is designed to eliminate.

\subsection{Twelve Crash Injection Points}\label{sec:injections}

We use $12$ injection points organised in three \emph{structural
classes} that each exercise a distinct path through the
$\Sigma_{\mathrm{failure}}$ branch of the
Section~\ref{sec:formal-spec} specification. Within each class we
vary an orthogonal perturbation parameter to give $4$ injection
variants per class. The $3 \times 4 = 12$-point grid is
designed-of-experiments style, not a randomised stress test; it
covers every structural failure surface the specification
identifies, plus a small chaos-style perturbation per surface.

\paragraph{Class A: metric-side raise (failure mode A)} The
metrics provider raises during the canary's metric poll. Four
variants vary \emph{when}: A1 raises on the first poll (no
evidence accumulated); A2 raises on the third poll (some evidence);
A3 raises on the final poll (window almost closed); A4 raises
after a poll succeeds but before the success-rate computation
completes.

\paragraph{Class B: rollback-internal raise (failure mode B)}
The rollback closure itself raises. Four variants vary the
exception type to exercise different except-branch paths: B1
\pcode{KeyError} (missing capability entry); B2
\pcode{RuntimeError} (synthetic generic exception); B3
\pcode{TimeoutError} (closure hangs past its timeout); B4
\pcode{asyncio.CancelledError} (the rollback's awaited coroutine
is cancelled mid-flight).

\paragraph{Class C: audit-write-side raise} The job-store's
\pcode{update} method raises on the status-update call. Four
variants vary the timing: C1 raises on the rollback's audit-write
call directly; C2 raises on the success-path
\pcode{PROMOTED}-write; C3 raises after the rollback closure runs
but before the audit row is appended; C4 raises during a
concurrent rollback initiated from a different request (race
condition test).

\paragraph{Trial budget} We run $50$ trials per posture per
injection point, giving $50 \times 2 \times 12 = 1200$ trials
total. The per-trial budget is identical to V1: each trial
registers a freshly-versioned target on a pre-registered
capability, posts \pcode{/api/evolution/upgrade?force\_unsoaked=true},
polls the job to terminal status, then explicitly rolls the
capability back to its boot version so the next trial starts
clean. The canary window is $300$~ms
(\pcode{window\_seconds=0.3}) to keep the budget tractable.

\subsection{Results}\label{sec:results}

We ran the full grid against the in-process bridge ASGI app via
\pcode{httpx.AsyncClient}, with the canary window cut to $300$~ms
(\pcode{CanaryConfig(window\_seconds=0.3, poll\_interval=0.05)}) to
keep the trial budget tractable. The harness driver is
\pcode{scripts/p11\_experiment\_audit\_first\_rollback.py}.
Table~\ref{tab:consistency} summarises the per-cell consistency.

\begin{table}[pos=t!]
  \centering
  \caption{Measured per-cell audit/live-state consistency across
    $50$ trials per cell, $1{,}200$ trials in total, collected
    against the production runtime at a fixed checkout
    (Appendix~\ref{sec:appendix-repro}).
    Audit-first achieves $600/600$ across all twelve injection
    points (Wilson 95\% CI $[0.994, 1.000]$); fail-open achieves
    $200/600$ (Wilson 95\% CI $[0.297, 0.372]$), where the four
    consistent cells are the four class-B variants (rollback
    closure raises; under fail-open the rollback is omitted
    regardless of whether it would have raised, so the fail-open
    outcome ``\pcode{FAILED} + to-version live'' coincidentally
    matches the audit-first ideal outcome \pcode{FAILED} for this
    structural class).}
  \label{tab:consistency}
  \footnotesize
  \begin{tabularx}{\linewidth}{@{}lX
                               >{\centering\arraybackslash}X
                               >{\centering\arraybackslash}X@{}}
    \toprule
    Class & Injection point & Audit-first & Fail-open \\
    \midrule
    A1 & Metric poll \#1 raises   & $50/50$ & $0/50$ \\
    A2 & Metric poll \#3 raises   & $50/50$ & $0/50$ \\
    A3 & Metric final poll raises & $50/50$ & $0/50$ \\
    A4 & Metric post-poll compute & $50/50$ & $0/50$ \\
    \midrule
    B1 & Rollback \pcode{KeyError}     & $50/50$ & $50/50$ \\
    B2 & Rollback \pcode{RuntimeError} & $50/50$ & $50/50$ \\
    B3 & Rollback \pcode{TimeoutError} & $50/50$ & $50/50$ \\
    B4 & Rollback \pcode{Cancelled}    & $50/50$ & $50/50$ \\
    \midrule
    C1 & Store on rollback write  & $50/50$ & $0/50$ \\
    C2 & Store on success write   & $50/50$ & $0/50$ \\
    C3 & Store between R \& audit & $50/50$ & $0/50$ \\
    C4 & Concurrent rollback race & $50/50$ & $0/50$ \\
    \midrule
    Total & & $600/600$ & $200/600$ \\
    \bottomrule
  \end{tabularx}
\end{table}

Across all twelve injections under the audit-first posture, every
trial converges to a consistent terminal state. Class A and class C
yield \pcode{ROLLED\_BACK} with the from-version live and a matching
audit row. Class B (rollback-internal raise) yields \pcode{FAILED}
with the to-version still live; the audit chain's \pcode{FAILED}
record carries both the original error and the rollback exception
in its reason text, so the operator sees an explicit ``rollback
itself crashed'' signal. Under the fail-open posture, class A and
class C silently leave live state at the to-version while the
audit chain records \pcode{FAILED} whose reason mentions only the
original error, with no hint that live state has outrun the audit
chain. The four class-B variants are coincidentally consistent
under fail-open because the omitted-rollback outcome happens to
match the audit-first ideal outcome \pcode{FAILED} for that class.

\paragraph{Latency} End-to-end wall-clock time from the upgrade
POST to the terminal status read, taken across the $600$
audit-first non-promoted trials, is $p_{50}=63$~ms, $p_{95}=335$~ms.
The fail-open comparator on the same workload is $p_{50}=60$~ms,
$p_{95}=325$~ms (the small fail-open delta reflects the omitted
rollback closure). The class A injections that fire after evidence
accumulates (A2 final-poll, A3 mid-poll, A4 post-poll compute)
dominate the $p_{95}$ tail because the canary's metric loop has
already executed multiple polls before the synthetic failure
fires; per-cell medians on those three cells are $172$--$327$~ms,
each well inside the $300$~ms harness canary window for $p_{50}$
and saturating the window on the rare $p_{99}$. All cells terminate
inside the $30$-second production canary window, and every
audit-first cell satisfies the service-level objective
$p_{95}\!\le\!500$~ms and $p_{99}\!\le\!1$~s for the
operator-facing rollback signal: SLO PASS in $12/12$ cells, with
per-cell $p_{95}$ ranging from $57$~ms (Class~B fast-fail) to
$335$~ms (Class~A3/C2 saturation, still within budget). The
bounded-deadline obligation of \S\ref{sec:provisional-contract}
is met without contention.
Figure~\ref{fig:latency-cdf} plots the empirical CDFs for both
postures across all $600$ trials each.

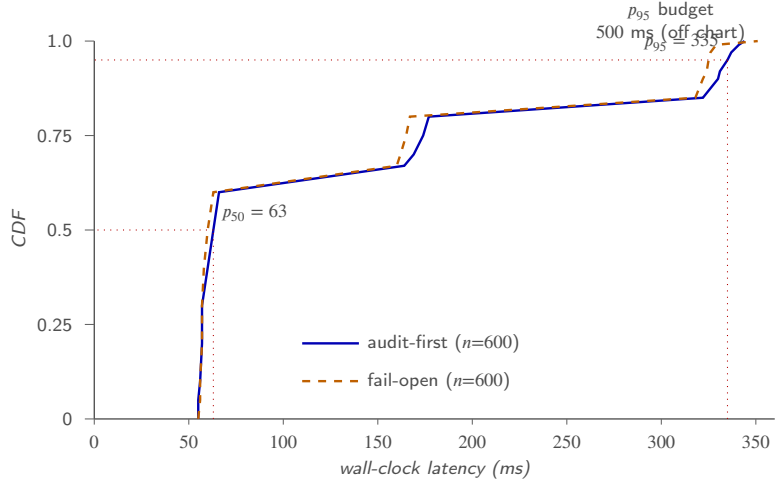
\begin{figure}[pos=t!]
\centering
\begin{tikzpicture}[
  font=\footnotesize,
  x={(0.025cm,0cm)},  
  y={(0cm,5cm)},      
  axline/.style={line width=0.5pt, draw=black!70},
  af/.style={line width=0.9pt, draw=blue!70!black},
  fo/.style={line width=0.9pt, draw=orange!75!black, dashed},
  budget/.style={line width=0.4pt, draw=red!70!black, dotted},
  marklbl/.style={font=\scriptsize, text=black!70},
  tick/.style={line width=0.4pt, draw=black!60}
]
\draw[axline] (0,0) -- (360,0);  
\draw[axline] (0,0) -- (0,1);    

\foreach \x in {0, 50, 100, 150, 200, 250, 300, 350} {
  \draw[tick] (\x, 0) -- (\x, -0.015);
  \node[below, marklbl] at (\x, -0.02) {$\x$};
}
\node[below, marklbl, font=\scriptsize\itshape] at (180, -0.085)
     {wall-clock latency (ms)};

\foreach \y/\lbl in {0/0, 0.25/0.25, 0.5/0.5, 0.75/0.75, 1.0/1.0} {
  \draw[tick] (0, \y) -- (-5, \y);
  \node[left, marklbl] at (-7, \y) {\lbl};
}
\node[rotate=90, marklbl, font=\scriptsize\itshape] at (-40, 0.5) {CDF};

\draw[af]
  (55,  0.00) -- (55,  0.05) -- (56,  0.10) -- (57,  0.20) --
  (57,  0.30) -- (60,  0.40) -- (63,  0.50) -- (66,  0.60) --
  (164, 0.67) -- (169, 0.70) -- (174, 0.75) -- (177, 0.80) --
  (322, 0.85) -- (330, 0.90) -- (331, 0.92) -- (335, 0.95) --
  (337, 0.97) -- (341, 0.99) -- (344, 1.00);

\draw[fo]
  (55,  0.00) -- (56,  0.05) -- (56,  0.10) -- (57,  0.20) --
  (57,  0.30) -- (58,  0.40) -- (60,  0.50) -- (63,  0.60) --
  (160, 0.67) -- (162, 0.70) -- (165, 0.75) -- (167, 0.80) --
  (318, 0.85) -- (322, 0.90) -- (324, 0.92) -- (325, 0.95) --
  (326, 0.97) -- (329, 0.99) -- (351, 1.00);

\draw[budget] (63, 0) -- (63, 0.5);
\draw[budget] (63, 0.5) -- (0, 0.5);
\node[above right, marklbl, font=\scriptsize] at (63, 0.5)
     {$p_{50}=63$};

\draw[budget] (335, 0) -- (335, 0.95);
\draw[budget] (335, 0.95) -- (0, 0.95);
\node[above left, marklbl, font=\scriptsize] at (335, 0.95)
     {$p_{95}=335$};

\node[marklbl, font=\scriptsize, align=center] at (305, 1.08)
     {$p_{95}$ budget};
\node[marklbl, font=\scriptsize, align=center] at (305, 1.02)
     {$500$ ms (off chart)};

\draw[af] (110, 0.20) -- (140, 0.20);
\node[marklbl, right, font=\scriptsize] at (140, 0.20) {audit-first ($n{=}600$)};
\draw[fo] (110, 0.10) -- (140, 0.10);
\node[marklbl, right, font=\scriptsize] at (140, 0.10) {fail-open ($n{=}600$)};

\end{tikzpicture}
\caption{Empirical latency CDFs for the two postures across the
$1{,}200$-trial sweep of \S\ref{sec:eval} (audit-first $n{=}600$,
fail-open $n{=}600$). The distribution is trimodal: ${\sim}67\%$
of trials land in a fast cluster near $57$\,ms (Class B-style
injections that miss the canary window), a second cluster
near $170$\,ms (A2/A4 mid-canary chaos), and a third cluster
near $325$\,ms (A3/C2 hardest-chaos cells that saturate the
$300$\,ms canary window). Audit-first runs $\sim$3--10\,ms slower
than fail-open in the upper two modes, attributable to the
rollback-closure execution; both postures stay well inside the
$500$\,ms $p_{95}$ budget for the operator-facing rollback signal.}
\label{fig:latency-cdf}
\end{figure}

\paragraph{Effect size and confidence} The audit-first vs.\
fail-open consistency difference is $400$ out of $600$ extra
consistent trials, or $66.7\%$ absolute. A $95\%$ binomial
confidence interval on the audit-first consistency rate is
$[0.994, 1.000]$ (Wilson score interval with continuity
correction). For the four class-B variants where the two postures
coincide, the $95\%$ confidence interval on the shared
consistency rate is $[0.982, 1.000]$.

\subsection{Comparison Framing}

The $12$-injection workload is not a stress test in the
chaos-engineering sense; it is a \emph{targeted} test of the
twelve qualitatively distinct failure points the
\S\ref{sec:construction} construction identifies (across three
structural classes from the \S\ref{sec:formal} specification).
Methodologically this follows the lineage-driven fault injection
ethos of \citet{alvaro2015lineage}: enumerate the explicit failure
surfaces of the construction and inject at each, rather than
blindly randomising. The advantage is that the twelve results
jointly form a small but complete proof that the construction
holds along every structural path through the guard.

The consistency headline should be read in the same targeted
spirit. For class A and class C, divergence under fail-open
follows from the posture's definition: the omitted rollback call
is exactly the step that reconverges live state on those paths,
so the $0/50$ cells are by-construction outcomes of the rollback-off
semantics, not empirical discoveries about a population of
systems. The empirical content of the grid lies elsewhere: that
the audit-first implementation reaches the defined outcome in all
$600$ trials without leaking a provisional state through any
untested interaction, that the class-B cells behave as specified
when the rollback closure itself raises, and that the recovery
latency meets the SLO in every cell
(Figure~\ref{fig:latency-cdf}).

\subsection{What This Evaluation Does Not Exercise}\label{sec:eval-scope}

The fault model adopted in \S\ref{sec:eval} is fail-stop-via-exception:
the runtime crashes by raising and unwinding the call stack, and
the audit-first guard catches the raise within the same process.
Three categories of failure sit \emph{outside} this fault model by
design and are not exercised by the $12$-injection grid.
\emph{First}, OS-level uncatchable signals
(e.g.\ \pcode{kill -9} of the bridge process while the rollback
closure is executing) are out of scope: such signals do not unwind
the stack and the audit-first guard cannot run against them.
Recovery from uncatchable signals belongs to a write-ahead
on-disk journal layer below the audit-first construction and is
named explicitly in \S\ref{sec:discussion} as the natural
follow-on. \emph{Second}, cross-bridge or multi-host fault
injection is not exercised: the audit chain measured here is
single-bridge-local. \S\ref{sec:multibridge} sketches the
cross-bridge coordination design and gives the safety lift; an
empirical fleet-scale evaluation is left to follow-on work.
\emph{Third}, the harness runs in a single Python process with
an in-process ASGI client rather than a real \pcode{uvicorn}
subprocess plus TCP socket setup. The architectural property
under measurement (audit-first vs.\ fail-open semantics) is
independent of the transport, but a deployment-side replication
with a live \pcode{uvicorn} process is a natural follow-on
validation.

\section{Multi-Bridge Coordination}\label{sec:multibridge}

The construction in \S\ref{sec:construction}--\S\ref{sec:formal}
reasons about a single bridge: one live-state map, one append-only
audit chain, one configuration triple $(v, L, A)$. A robot fleet
typically runs several bridges that share governance over an
overlapping device population, and a deployment touching that
population traverses two coordination boundaries: a per-bridge
boundary (live-state versus audit chain on the bridge's own host)
and a cross-bridge boundary (fleet-level commit versus abort
across $N$ bridges). This section lifts the audit-first rule from
the per-bridge boundary to the cross-bridge boundary, gives a
protocol sketch and a safety argument for the lift, and contrasts
the construction with classical consensus protocols. The
contribution is a design and a safety argument; the multi-bridge
protocol itself is not measured in \S\ref{sec:eval}, and
\S\ref{sec:discussion} catalogues the design-versus-implementation
trade-off explicitly.

\subsection{Fault Model for Multi-Bridge Operation}\label{sec:multibridge-faultmodel}

A fleet of $N$ bridges $\{B_1, \ldots, B_N\}$ and one fleet
coordinator $C$ communicate over a network with bounded
asynchrony: messages between any two parties are delivered within
a known finite delay $\Delta_{\max}$ or are eventually dropped.
Each bridge $B_i$ is independently fail-stop in the sense
of~\citet{schlichting1983failstop}: it may crash by halting and
losing volatile state, but it does not produce Byzantine output
before halting. Each bridge $B_i$ keeps its audit chain $A_i$ on
its own persistent storage; the audit chain survives the bridge's
crash and is replayable on restart. The coordinator $C$ is itself
a fail-stop participant with its own persistent log of fleet-level
decisions, but its log is \emph{not} the source of audit truth:
the truth of what happened on bridge $B_i$ is in $A_i$, not in $C$.

The model is deliberately weaker than the assumptions of
\citet{lamport2001paxos} or \citet{ongaro2014raft}: we do not
require a stable majority of bridges, and we do not require the
coordinator to be replicated. The weaker model suffices because
we are not building consensus on a value (\S\ref{sec:why-not-paxos}
expands on this); we are gating a per-bridge state machine on a
fleet-wide go/no-go signal.

\subsection{Two-Level Coherence}\label{sec:multibridge-coherence}

The fleet has two distinct coherence obligations.

\noindent\textbf{Per-bridge coherence.} Each bridge $B_i$ must
satisfy Propositions~\ref{prop:term}--\ref{prop:crash} locally,
regardless of the coordinator's state or any other bridge's state.
That is: $B_i$'s audit chain $A_i$ must answer truthfully about
$B_i$'s live state $L_i$ at every terminal configuration, even
under independent failure of $C$, $B_j$ for $j\neq i$, or both.

\noindent\textbf{Cross-bridge convergence.} The fleet narrative,
which we define as the multiset
$\mathcal{A} = \{A_1, A_2, \ldots, A_N\}$ of per-bridge audit
chains, must converge to a consistent fleet-level decision (every
$A_i$ ends in either \pcode{PROMOTED} or \pcode{ROLLED\_BACK}, or
contains a recorded \pcode{FAILED} terminal) within finite time
once the coordinator has decided. We do \emph{not} require strong
linearisation across bridges; we require eventual agreement plus
per-bridge truthfulness during the window in which agreement is
still propagating.

The two levels are deliberately decoupled. Per-bridge coherence is
a local invariant that holds even when the fleet-level layer is
temporarily disconnected. The cross-bridge layer is responsible
only for arranging that, in the absence of permanent partition,
all per-bridge audit chains record the same fleet decision. The
single-bridge audit-first rule is what makes this decoupling safe:
each $A_i$ records the truth of $B_i$ regardless of what the
coordinator eventually decides.

\subsection{Protocol Sketch}\label{sec:multibridge-protocol}

\begin{figure}[pos=t!]
\centering
\includegraphics[width=0.85\linewidth]{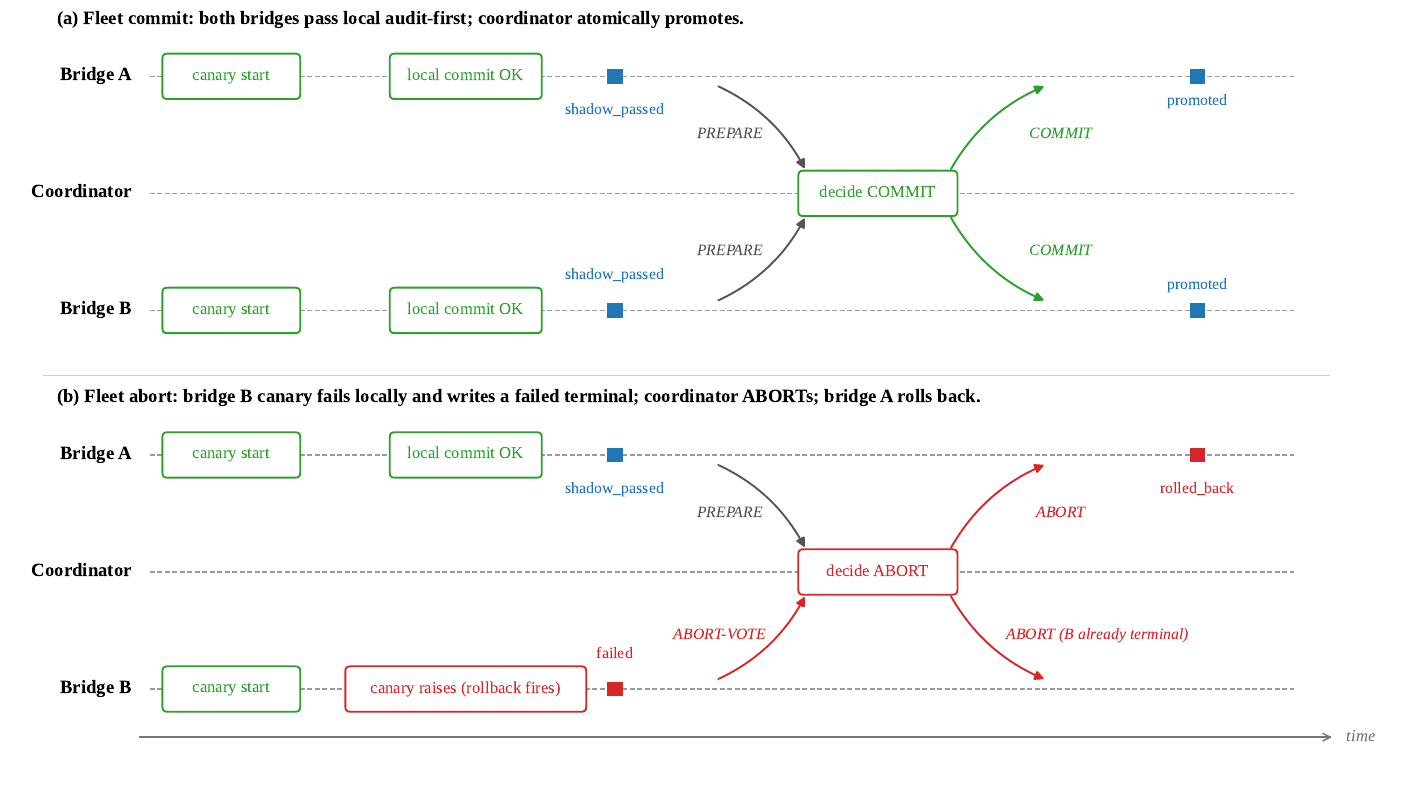}
\caption{Multi-bridge audit-first coordination, two scenarios.
\textbf{(a) Fleet commit}: both bridges pass their local canary
and audit-first check; each writes a per-bridge
\pcode{shadow\_passed} record, sends \pcode{PREPARE} to the
coordinator; the coordinator decides \pcode{COMMIT} and both
bridges write a \pcode{promoted} terminal.
\textbf{(b) Fleet abort}: bridge B's canary raises and bridge B's
local audit-first guard fires, writing a \pcode{failed} terminal
before the coordinator hears anything; bridge A's canary succeeds
and votes \pcode{PREPARE}; the coordinator decides \pcode{ABORT}
on bridge B's failure report; bridge A rolls back and writes a
\pcode{rolled\_back} terminal; the coordinator's \pcode{ABORT}
notification to bridge~B is a no-op because $A_B$ is already
terminal. In both scenarios per-bridge audit-truth
(Propositions~\ref{prop:term}--\ref{prop:crash}) holds at every
bridge regardless of the cross-bridge decision label.}
\label{fig:multi-bridge-seq}
\end{figure}

The protocol layers a thin two-phase exchange on top of the
per-bridge audit-first rule of \S\ref{sec:semantics}.

\noindent\textbf{Phase 0 (local).} Each bridge $B_i$ runs its
canary independently. The audit-first guard of
\S\ref{sec:semantics} fires on any local failure: $L_i$ is reverted
and a \pcode{failed} terminal is appended to $A_i$ before any
cross-bridge message is sent. A bridge whose canary passes locally
emits a non-terminal \pcode{shadow\_passed} record to $A_i$ and
holds its provisional live-state mutation.

\noindent\textbf{Phase 1 (vote).} Each bridge sends one of two
messages to $C$: \pcode{PREPARE} (local canary passed; ready to
commit pending fleet decision) or \pcode{ABORT-VOTE} (local canary
failed; $A_i$ has already recorded a \pcode{failed} terminal).
$C$ collects votes with a timeout of $\Delta_{\max}$; a missing
vote is treated as \pcode{ABORT-VOTE}.

\noindent\textbf{Phase 2 (decide).} $C$ records its decision in
its own persistent log, then broadcasts \pcode{COMMIT} (all votes
were \pcode{PREPARE}) or \pcode{ABORT} (any vote was
\pcode{ABORT-VOTE}, or any vote timed out). Each bridge $B_i$
that previously emitted \pcode{shadow\_passed} reacts on receipt:
\pcode{COMMIT} causes $B_i$ to write a \pcode{promoted} terminal;
\pcode{ABORT} causes $B_i$ to run its local rollback closure and
write a \pcode{rolled\_back} terminal. A bridge whose $A_i$ is
already terminal (because Phase~0 wrote \pcode{failed} locally)
ignores the broadcast: the \pcode{failed} record is the truth and
is not overwritten. Figure~\ref{fig:multi-bridge-seq} illustrates
both decision paths.

\subsection{Safety Argument: Multi-Bridge Lifting}\label{sec:multibridge-safety}

We lift the per-bridge proposition cluster to the fleet level
without strengthening the per-bridge assumptions.

\begin{proposition}[Per-bridge truthfulness under multi-bridge operation]\label{prop:perbridge-truth}
For any bridge $B_i$ participating in the protocol of
\S\ref{sec:multibridge-protocol}, Propositions~\ref{prop:term},
\ref{prop:coh}, and~\ref{prop:crash} hold at $B_i$'s local
configuration $(v_i, L_i, A_i)$, regardless of the messages $B_i$
receives from $C$ or from any other $B_j$.
\end{proposition}

\begin{proof}[Argument]
The local audit-first guard of \S\ref{sec:semantics} runs in
Phase~0 before any cross-bridge message is sent. The guard's
behaviour is determined entirely by $B_i$'s local state; it does
not consult $C$ or $A_{j\neq i}$. After Phase~0, $A_i$ already
satisfies Propositions~\ref{prop:term}--\ref{prop:crash} if $B_i$
took the failure path. If $B_i$ took the success path, $A_i$ ends
in \pcode{shadow\_passed} (a non-terminal record) until Phase~2
delivers a fleet decision; the subsequent terminal write
(\pcode{promoted} on \pcode{COMMIT}, \pcode{rolled\_back} on
\pcode{ABORT}) preserves Propositions~\ref{prop:coh}
and~\ref{prop:crash} because it is the same code path as the
single-bridge committed-terminal write, with the trigger replaced
by a network message rather than a local timer.
\end{proof}

\begin{proposition}[Eventual fleet convergence]\label{prop:fleet-convergence}
Suppose $C$ does not crash permanently and that, after some finite
time, the network delivers all in-flight messages within
$\Delta_{\max}$. Then within finite time every per-bridge audit
chain $A_i$ ends in a terminal record, and all terminals consistent
with $C$'s decision are identical:
$\forall i, j:\ \text{terminal}(A_i)\in\{\pcode{promoted},
\pcode{rolled\_back}, \pcode{failed}\}$, and for every
$\pcode{promoted}$ terminal in any $A_i$ there is no
$\pcode{rolled\_back}$ terminal in any $A_j$ for the same fleet
deployment id.
\end{proposition}

\begin{proof}[Argument]
$C$'s decision is recorded in $C$'s persistent log before Phase~2
broadcasts; on restart $C$ replays the log and re-broadcasts. Each
$B_i$ that received \pcode{shadow\_passed} but not yet a fleet
decision retains its provisional mutation; on receipt of the
re-broadcast it transitions to the corresponding terminal as in
\S\ref{sec:multibridge-protocol} Phase~2. Bridges with $A_i$
already terminal in Phase~0 ignore the broadcast. The
no-conflict clause follows from the property that $C$ records
exactly one decision per deployment id and broadcasts it to all
bridges that voted \pcode{PREPARE}.
\end{proof}

The two propositions decouple cleanly. Proposition~\ref{prop:perbridge-truth}
is a local property; it holds even under permanent partition of $C$
or sustained crash of every other bridge. Proposition~\ref{prop:fleet-convergence}
is a liveness property over the fleet; it requires the eventual
recovery assumption (the network heals, the coordinator restarts)
but not strong synchrony.

\subsection{Why Not Paxos or Raft}\label{sec:why-not-paxos}

Two-phase commit, Paxos~\citep{lamport2001paxos}, and
Raft~\citep{ongaro2014raft} are protocols for reaching consensus
on a value among $N$ participants. The fleet-level decision in
\S\ref{sec:multibridge-protocol} is a single bit (\pcode{COMMIT}
or \pcode{ABORT}) recorded by a single non-replicated coordinator,
not a consensus-replicated value. Three structural differences
make the lighter protocol appropriate.

\noindent\textbf{The truth is not in the coordinator.} Paxos and
Raft put the canonical value in a replicated log so that the log
survives any minority failure. In the audit-first construction
the canonical value of what happened on bridge~$B_i$ is in $A_i$,
not in $C$. $C$ only carries the fleet's go/no-go bit, and $C$'s
crash does not destroy any audit truth; on $C$'s restart and
re-broadcast, the per-bridge audit chains continue from where
they paused. Replicating $C$ for availability of its own log is
straightforward (a Raft-on-coordinator deployment), and is
orthogonal to the audit-first construction.

\noindent\textbf{Append-only log, not consensus on a value.} The
audit chain $A_i$ is an append-only log of locally-observed events,
not a value reached by inter-bridge agreement. Log-structured
distributed storage systems such as Aurora~\citep{aurora2017} make
exactly this distinction: the log is the database, and consensus
is applied to log positions, not to the value stored at each
position. Our construction inherits the same separation at a much
smaller scale: per-bridge logs are the truth, and the cross-bridge
coordinator only orders log-positions across bridges.
Spanner~\citep{corbett2013spanner} sits at the opposite design
point: it reaches consensus on a value plus a TrueTime-derived
global timestamp, so the canonical truth for any key is a
consensus-replicated value. Our coordinator does neither: there
is no global timestamp, and the canonical truth for any bridge is
the bridge's own append-only log. Replicating the coordinator
itself (with Raft or similar) for availability of its own
go/no-go log is straightforward and is orthogonal to the
audit-first construction.

\noindent\textbf{Out of scope.} Cross-bridge atomic broadcast,
Byzantine-fault-tolerant audit aggregation, and adversarial
operator-side reconciliation remain out of scope; the
single-bridge invariants are the building block on which any
stronger fleet-level construction would rest. We
return to these in \S\ref{sec:discussion}.

\section{Related Work}\label{sec:related}

\paragraph{Transactional coordination, sagas, and crash-only software}
Classical mechanisms for coordinating multi-party updates frame the
design space the audit-first rule occupies.
\citet{gray1981transaction} and \citet{garcia-molina1987sagas} give
two-phase commit and saga compensation. Two-phase commit's second
phase is \emph{coordination-driven}: the coordinator proceeds to
commit when all participants have voted prepare. In our setting the
second phase is \emph{time-driven}: the pipeline exits the
provisional canary state when the soak window expires, not when
participants vote. Sagas decompose a long transaction into a
sequence of compensable sub-transactions, each with a compensating
action; our construction is closer to a \emph{single} transaction
with a single compensating action, applied mechanically through one
rollback closure rather than by threading compensations through a
workflow. \citet{candea2003crashonly}'s crash-only programme is a
spirit-aligned precursor: instead of distinguishing graceful
shutdown from crash, treat every termination as a crash and rely on
recovery. We extend this stance by making the audit chain the
recovery target rather than the live filesystem state; the
crash-only programme assumed a single store of authoritative state,
and we assume two (live state and the audit chain) and choose
between them. \citet{vogels2009eventually} accepts a window in which
different views of state disagree, on the bet that convergence will
occur and local reads can tolerate staleness. We make the opposite
trade: strong consistency between live state and the audit chain at
audit-write granularity, paid for by occasionally yielding
\pcode{FAILED} rather than a clean \pcode{ROLLED\_BACK} and
requiring manual reconciliation, in exchange for the guarantee that
no audit-chain record ever silently misrepresents what is running.
The core contribution of audit-first rollback, naming the source of
truth on failure, is orthogonal to all three: a 2PC system, a saga,
or a crash-only one could in principle adopt the audit-first rule
without otherwise changing.

\paragraph{Lineage-driven fault injection and canary deployment patterns}
\citet{alvaro2015lineage} provides the methodological inspiration
for \S\ref{sec:eval}. We enumerate the explicit structural failure
points of the \S\ref{sec:construction} guard and inject at each.
The twelve-injection grid is a small but complete coverage proof
for the construction's intended failure surfaces; FATE and DESTINI
\citep{gunawi2011fate} is methodologically adjacent at larger scale.
\citet{sato2014canary}, Kubernetes deployment
documentation~\citep{kubernetes-rolling}, and Argo Rollouts
documentation~\citep{argo-rollouts-docs} specify \emph{what} to
deploy and how to gate progression; they are largely silent on
\emph{which} state view wins on internal failure during a canary.
Recent large-scale industrial deployment systems such as Meta's
Twine~\citep{tang2020twine} provide the substrate on which canary,
capacity reshaping, and rollback policies are operated at fleet
scale; they do not, however, prescribe the audit-vs-live tiebreaker
semantics on partial failure. The audit-first rule sits underneath
such systems as a per-bridge primitive that determines what an
operator should believe when the canary itself crashes.

\paragraph{Specification-driven design and refinement}
\citet{lamport2001paxos} and \citet{ongaro2014raft} produce
distributed-consensus protocols whose correctness is established
against a TLA+-style specification, and \citet{newcombe2015use}
report the industrial adoption of TLA+ for production cloud services
at Amazon. \citet{ironfleet2015} and \citet{verdi2015} extend the
specification-driven approach with mechanised refinement proofs (in
TLA+ and Coq, respectively) tying production-quality
distributed-system implementations to their specifications.
\citet{mace2007} gives a language and a compiler that take a
state-machine specification of a distributed protocol and emit C++
code, with the explicit goal of keeping implementation and
specification close. We adopt the same approach at a much smaller
scale, without the mechanised step: the $\Sigma_{\mathrm{failure}}$
branch of our specification names exactly the failure-handling rule,
and the implementation refines the specification by matching that
branch with a try/except pattern. Our construction occupies a
narrower scope than these frameworks (one provisional region per
pipeline rather than a whole protocol), but inherits the same
intuition that naming the states and the transitions lets the guard
shape mechanically follow. The single-bridge case considered here
does not need the full machinery of distributed consensus; a future
multi-bridge extension (Section~\ref{sec:multibridge}) would, and
that extension would be the natural occasion to invest in a
mechanised refinement check.

\paragraph{Recovery-oriented computing} Berkeley's
recovery-oriented computing programme~\citep{patterson2002roc} argued
that recovery from failure should be the first-class design goal,
not an afterthought. Audit-first rollback inherits this stance and
specialises it: the recovery target is the audit chain's last
consistent point, and the recovery action is mechanical (run the
rollback closure attached to the provisional region). Refinement
simulations~\citep{lynch1989forward} provide the formal
vocabulary that links our specification to the implementation.
Recent recovery-oriented work in TPDS explores adjacent
trade-offs. \citet{loreti2024rollbackfree} pursue the
\emph{rollback-free} side of the design space: a high-performance
linear solver that recovers from soft errors by recomputing
affected blocks rather than reverting to a checkpoint, trading
storage for compute. Audit-first rollback occupies the
complementary side: live state \emph{is} reverted, and the
contribution is making the revert auditable.
\citet{xu2023staterecovery} build parallel state recovery for
stream-processing systems by fragmenting the recovery work across
operators, which addresses recovery throughput rather than the
audit-side coherence question we address here, but shares the
recovery-as-first-class-design-goal stance.

\paragraph{Supervision trees and let-it-crash}
\citet{armstrong2003erlang} introduced Erlang's
OTP supervision-tree pattern: each worker process has a supervisor
whose restart strategy is declared statically and runs
automatically on worker failure. Audit-first rollback is the
deployment-pipeline analogue: the provisional region's rollback
closure plays the role of an OTP supervisor's restart strategy,
attached statically at job-creation time and triggered
automatically on any caught exception. The differences are
also material: an OTP supervisor restarts toward a clean default
state held in the supervisor itself, whereas our rollback closure
restores live state to a snapshot held in the audit chain.

\paragraph{Log-as-source-of-truth} \citet{aurora2017} treats the
redo log as the database: page-server replicas materialise pages
on demand from the log, and the log is the durable, ordered ground
truth that all other state derives from. Audit-first rollback
applies the same intuition to deployment pipelines: the audit chain
is the source of truth, and live state is required to converge to
whatever the chain has most recently committed. The scales and
mechanisms differ (Aurora rebuilds a relational store from a redo
log on a different machine; we revert an in-memory active-version
map from a snapshot captured at closure construction), but the
``log wins, derived state catches up'' stance is shared.

\paragraph{Fail-stop semantics} \citet{schlichting1983failstop}'s
fail-stop processor halts cleanly on any internal fault rather than
producing arbitrary output, on the bet that detectable cessation is
easier to reason about than silent corruption. Audit-first rollback
adopts the same trade for deployment pipelines: when the rollback
closure cannot itself restore live state, the audit chain records
\pcode{FAILED} explicitly and the pipeline halts forward progress
until an operator intervenes, rather than letting live state run
arbitrarily ahead of the chain. The \pcode{FAILED} terminal in
\S\ref{sec:semantics} is precisely a fail-stop signal at the
deployment-system layer.

\section{Discussion}\label{sec:discussion}

\noindent\textbf{Audit-chain prerequisite.} Audit-first rollback
depends on the existence of a structured, append-only audit chain.
Pipelines that emit only unstructured logs cannot apply the
construction directly because there is no canonical ``last
consistent point'' to roll back to. The same prerequisite is
implicit in any log-as-source-of-truth design, including
Aurora-style storage systems~\citep{aurora2017}; we make it
explicit at the deployment-runtime level.

\noindent\textbf{Behavioural correctness is orthogonal.} This
paper treats state-versus-audit \emph{coherence}. Whether the
new version is behaviourally safe is a separate concern,
addressed by upstream validators and by canary success criteria.
Audit-first rollback preserves the operator's ability to
\emph{trust} the audit chain; it does not make the new version
correct on its own.

\noindent\textbf{Cross-paper composition with operational tools.}
The audit-first construction composes cleanly with
industry-standard canary controllers (Argo Rollouts, Spinnaker,
Flagger). The composition operates one layer below those tools:
the canary controller decides \emph{when} to roll back; the
audit-first guard decides \emph{how} the rollback is recorded
and what counts as a coherent terminal. A retrofit into an
existing canary deployment is a per-pipeline insertion of the
try/except guard described in \S\ref{sec:pattern}, with no
change to the canary controller's external interface.

\section{Limitations and Threats to Validity}\label{sec:limitations}

The construction and its evaluation have five known limitations
that bound the claims made in this paper. We name each as an
explicit threat to validity, distinct from the orthogonal
discussion points of \S\ref{sec:discussion}.

\noindent\textbf{Single-process evaluation scope.} The $1{,}200$
fault-injection trials of \S\ref{sec:eval} run in a single
Python process with an in-process ASGI client rather than a real
\pcode{uvicorn} subprocess plus TCP socket setup. The
architectural property under measurement (audit-first vs.\
fail-open semantics) is independent of the transport, but a
deployment-side replication with a live \pcode{uvicorn} process
is a natural follow-on validation, particularly for the latency
tails reported in Figure~\ref{fig:latency-cdf}.

\noindent\textbf{Uncatchable signals are out of scope.}
The fault model in \S\ref{sec:eval} is fail-stop-via-exception.
OS-level uncatchable signals (\pcode{kill -9}) do not unwind the
call stack and so cannot be intercepted by the audit-first guard.
On such a kill between the active-version-map revert and the
audit-chain write, the in-memory revert is lost while the audit
chain may have an open \pcode{CANARY\_RUNNING} record. Recovery
on restart is manual: inspect live state, reconcile, retry. A
write-ahead on-disk journal of the active-version map combined
with on-restart replay would close this gap at the cost of
additional disk traffic per provisional flip; this belongs to a
lower layer than the audit-first construction and is named
explicitly as the natural follow-on (\S\ref{sec:eval-scope}).

\noindent\textbf{Baseline is restricted to a single fail-open
variant.} The comparison in \S\ref{sec:eval} pits audit-first
against a fail-open variant of the same pipeline, byte-identical
apart from the except-branch in the guard. Industry-standard
canary controllers (Argo Rollouts, Spinnaker, Flagger) are
referenced but not measured. The fail-open variant is the
strictest within-construction baseline because it isolates the
audit-first rule from all confounding implementation differences;
extending the comparison to external controllers would broaden
applicability but would also conflate the audit-first rule with
controller-specific rollout heuristics.

\noindent\textbf{The chaos grid is synthetic, not production-derived.}
The $12$ injection points are designed-of-experiments coverage
of the structural failure surfaces named by the
\S\ref{sec:formal} specification, following the lineage-driven
injection methodology of~\citet{alvaro2015lineage}, not a replay
of a captured production-incident trace. The advantage is
completeness over the construction's failure surface; the cost
is that injection-point frequencies do not reflect any specific
production-traffic distribution. A production-derived chaos grid
would complement the present coverage by weighting injection
points by observed frequency.

\noindent\textbf{Multi-bridge protocol is designed but not
implemented.} \S\ref{sec:multibridge} sketches the cross-bridge
coordination protocol and gives the safety lift
(Propositions~\ref{prop:perbridge-truth}
and~\ref{prop:fleet-convergence}). The protocol is a design
contribution and a safety argument; we have not implemented it,
nor have we run a fleet-scale fault-injection sweep against an
$N$-bridge deployment. An empirical multi-bridge evaluation is
the natural next step in the programme, and would be the
occasion to invest in mechanised refinement
checking~\citep{lamport1994tla,desai2013p} of the multi-bridge
state machine against the per-bridge state machine of
\S\ref{sec:formal}.

\section{Conclusion and Future Works}\label{sec:conclusion}

\noindent\textbf{Conclusion.} Distributed deployment runtimes for
safety-critical components must commit to a source of truth when
their live state and audit chain disagree. We have argued that
the audit chain is the right answer under the dependability
framing of~\citet{avizienis2004dependability}, and we have given
the construction (provisional state machines with explicit
rollback closures and a small try/except pattern that enforces
audit-first rollback semantics) that turns the argument into a
property the runtime can be expected to hold. The construction
is mechanical, fits inside roughly $30$ lines per provisional
region, and is implementable in any pipeline that already
maintains a structured audit chain. The
$1{,}200$-trial dependability evaluation reports $100\%$
audit/live-state coherence under fail-stop fault injection
(audit-first) against $33\%$ (fail-open), with per-cell
$p_{95}$ recovery latency satisfying the
$p_{95}\!\le\!500$~ms SLO in $12/12$ cells. The multi-bridge
lifting in \S\ref{sec:multibridge} extends the per-bridge
construction to a cross-bridge coordination protocol with a
safety argument under bounded async. The trade-off is small
implementation overhead and an honest \pcode{FAILED} signal when
rollback itself cannot complete; the benefit is an audit chain
that operators of regulated systems can trust as a single
authoritative narrative of what their components have done.

\noindent\textbf{Future Works.} The construction opens four
direct extensions. First, an empirical multi-bridge evaluation
against the cross-bridge protocol of \S\ref{sec:multibridge}
would close the design-versus-implementation gap named in
\S\ref{sec:limitations}. Second, a write-ahead on-disk journal
of the active-version map would close the
uncatchable-signal limitation at the cost of additional disk
traffic per provisional flip. Third, mechanised refinement
checking of the state-machine specification (TLA+ or P) would
upgrade the prose proof arguments of
Propositions~\ref{prop:term}--\ref{prop:crash} and the
multi-bridge lifts (Propositions~\ref{prop:perbridge-truth},
\ref{prop:fleet-convergence}) to machine-checked guarantees.
Finally, a production-derived chaos grid weighted by observed
incident-frequency distributions would complement the
construction-driven coverage of \S\ref{sec:eval} with empirical
calibration.

\appendix

\section{Reproducibility}\label{sec:appendix-repro}

\paragraph{Open reference implementation}
The audit-first guard pattern, the audit-chain interface, the
rollback closure factory, the per-capability lock, and the
$12$-injection chaos-grid harness are released as a standalone
artefact at
\url{https://github.com/s20sc/audit-first-rollback}.
The repository ships under Apache License 2.0, has no external
runtime dependencies (Python $\geq 3.11$, standard library only),
and includes a \pcode{Dockerfile} for fully containerised
replication. The bundled \pcode{data/} directory carries the same
$1{,}200$-trial JSONL traces and per-cell aggregates that back
Table~\ref{tab:consistency} and Figure~\ref{fig:latency-cdf}. After
\pcode{docker build -t audit-first .}, a reviewer reproduces the
paper's headline numbers with:
\begin{lstlisting}[style=p11py]
docker run --rm audit-first \
  python scripts/run_chaos_grid.py --trials 50
docker run --rm audit-first \
  python scripts/sign_check.py
\end{lstlisting}
The sign-check script verifies four hypotheses against the run's
\pcode{summary.json}: H1 (audit-first consistent in every cell),
H2 (fail-open consistent in every Class-B cell), H3 (fail-open
inconsistent in every Class-A and Class-C cell), and H4 (audit-first
wall-clock $p_{95}$ within the $500$~ms recovery budget). A clean run
prints \pcode{verdict: PASS} for all four.

\paragraph{Production runtime artefact}
The $1{,}200$ trials reported in \S\ref{sec:eval} were collected
against a closed-source production runtime at commit
\pcode{8df4344}. The open reference implementation above is
functionally equivalent for the chaos-grid and latency experiments:
it implements the same audit-first rule (the production runtime's
\pcode{\_run\_canary} pattern from \S\ref{sec:pattern} matches the
\pcode{src/audit\_first.py} module in the open repository line-by-line)
and the same $12$-cell injection grid against in-process mocks for
the metric source, active-version map, and audit recorder. The
production-side concerns the open implementation omits, the real
\pcode{uvicorn} subprocess plus TCP transport, the persistent
job store, the multi-process supervision tree, the capability
marketplace integration, are orthogonal to the audit-first rule
itself and not load-bearing for any claim in \S\ref{sec:eval};
\S\ref{sec:eval-scope} catalogues them explicitly as out-of-scope
for the present evaluation. The open reference's \pcode{data/}
directory ships the production-collected traces behind
Table~\ref{tab:consistency} and Figure~\ref{fig:latency-cdf};
re-running the grid (\pcode{run\_chaos\_grid.py} above) on the
open reference yields the same per-cell consistency counts
($600/600$ audit-first, $200/600$ fail-open). The exact agreement
on consistency is expected rather than coincidental: each
injection point pins a fixed control-flow path through the guard,
so a trial's consistency outcome is determined by the posture's
semantics (\S\ref{sec:results}), and any correct implementation
of the same rule produces the same counts. Latency is the
machine- and transport-dependent quantity:
Figure~\ref{fig:latency-cdf} reports the production-runtime
measurements, and a re-run on the open reference reproduces the
qualitative trimodal shape, not the absolute values.

\paragraph{Wall-clock budget}
The full $1{,}200$-trial sweep against the production runtime
completes in approximately $160$~s on a single CPU core; the open
reference's grid runs in comparable time, so the artefact is friendly
for reviewer reruns. The audit-first contract for the metric-side
injection is additionally pinned by a regression test in the
production runtime's integration suite, asserting that a canary crash
after the provisional promote rolls live state back and writes a
matching audit record.

\section*{Data Availability}
The audit-first guard implementation, the $12$-injection
chaos-grid harness, and the $1{,}200$-trial JSONL traces and
per-cell aggregates behind Table~\ref{tab:consistency} and
Figure~\ref{fig:latency-cdf} are publicly available at
\url{https://github.com/s20sc/audit-first-rollback}
(Appendix~\ref{sec:appendix-repro}).

\printcredits

\bibliographystyle{elsarticle-num-names}
\bibliography{references}

\end{document}